\documentclass[final]{MIT8X10M}
\slugtoks{Sparse Modeling in Practice}

\usepackage{color}
\usepackage{makeidx}
\usepackage{ifthen}
\usepackage{eucal}

\usepackage{graphicx}
\usepackage{psfrag}

\usepackage[fleqn]{amsmath}
\usepackage{amssymb}
\usepackage{amsbsy}
\usepackage{framed}

\usepackage[chapter]{algorithm}
\usepackage{algorithmicx}
\usepackage{algmatlab} 
\usepackage{algpseudocode} 
\usepackage{lscape} 

\usepackage[ruled,vlined]{algorithm2e} 

\usepackage{amsthm}
\newtheorem{theorem}{Theorem}[chapter]
\newtheorem{definition}[theorem]{Definition}

\newtheorem{corollary}[theorem]{Corollary}
\newtheorem{lemma}[theorem]{Lemma}
\newtheorem{remark}[theorem]{Remark}

\usepackage{url}

\usepackage{smallcaptions}

\usepackage{bm}

\newcommand{\mm}{\bm{M}}

\newcommand{\R}{\mathbb{R}}

\DeclareMathOperator*{\argmin}{argmin}

\usepackage[american]{babel}
\usepackage{amssymb,amsfonts}
\usepackage{amsthm}
\usepackage{amsmath}
\usepackage{graphicx}
\usepackage{hhline}
\usepackage{enumerate}
\usepackage{float}
\usepackage{xspace}
\usepackage{algorithm}
\usepackage{nicefrac,enumerate,paralist}
\usepackage{boxedminipage,color,url}
\usepackage{multirow}
\usepackage{amssymb,amsmath}
\usepackage{subfigure}
\usepackage{algorithmicx}

\usepackage{eqparbox}

\newcommand{\bs}[1]{\boldsymbol{{#1}}}
\renewcommand{\sp}[1]{{\rm H}_k({#1})}

\newcommand{\as}{\bs{\alpha}}
\newcommand{\A}{\bs{\Phi}}
\newcommand{\f}{\bs{f}}
\newcommand{\m}{M}
\newcommand{\n}{N}
\newcommand{\e}{\bs{e_M}}

\newcommand{\sas}{\bs{\alpha^*}}
\newcommand{\has}{\bs{\hat{\alpha}}}
\newcommand{\inner}[2]{\langle {#1},{#2}\rangle}

\renewcommand{\mm}{\m}
\renewcommand{\R}{\mathbb{R}}
\renewcommand{\P}{\bs{\rm P}}
\newcommand{\Q}{\bs{\rm Q}}
\newcommand{\simp}{\Xi}
\newcommand{\breg}{{\cal B}_{\F}}
\newcommand{\F}{{\cal R}}
\newcommand{\proj}{{\cal P}_\Omega}
\newcommand{\muse}{{\rm GAME} }

\newcommand{\Lo}{{\cal L}}

\mathchardef\ordinarycolon\mathcode`\:
\mathcode`\:=\string"8000
\begingroup \catcode`\:=\active
  \gdef:{\mathrel{\mathop\ordinarycolon}}
\endgroup

\newenvironment{abstract}
 {\begin{list}{}
    {\setlength{\leftmargin}{0em}
     \setlength{\rightmargin}{0em}
     \setlength{\itemsep}{0pt}
     \setlength{\topsep}{0pt}}
    \item[] \em }
 {\end{list}\par\vspace{8ex}\par}

\newenvironment{authors}
 {\vspace{-4ex}\begin{list}{}
    {\setlength{\leftmargin}{-.5\marginpartotal}
     \setlength{\rightmargin}{0pt}
     \setlength{\itemsep}{0pt}
     \setlength{\topsep}{0pt}}
    \item[] }
 {\end{list}\par\vspace{8ex}\par}

\def\AUname#1{\par\makebox[2.8in][l]{\bf #1}}
\def\AUemail#1{{\tt #1}}
\def\AUweb#1{}
\def\AUaffiliation#1{\\\emph{#1}}
\def\AUaddress#1{\\\emph{#1}\par\medskip}

\makeatletter
\renewcommand\appendix{\par%
  \setcounter{section}{0}%
  \setcounter{subsection}{0}%
  \section*{Appendix}%
  \def\thesection{\thechapter.\@Alph\c@section}%
  \def\thesubsection{\thechapter.\@Alph\c@subsection}}
\makeatother

\begin{document}
\raggedbottom

%

\renewcommand{\sp}[1]{{\rm H}_k({#1})}

\newcommand{\obs}{\f}
\newcommand{\sens}{\A}
\newcommand{\signal}{\bs{\alpha}}
\newcommand{\bestsignal}{\bs{\alpha}^\ast}

\newcommand{\dimension}{\n}
\newcommand{\numsam}{\m}
\newcommand{\sparsity}{k}
\newcommand{\xtrue}{\as^\ast}
\newcommand{\noise}{\bf{n}}

\newcommand{\constraint}{\mathcal{C}_\sparsity}
\newcommand{\constrainttwo}{\mathcal{C}_{2\sparsity}}
\newcommand{\structsub}{\mathcal{V}_{\constraint}}
\newcommand{\projection}{\mathcal{P}_{\constraint}}
\newcommand{\apprprojection}{\mathcal{P}_{\constraint}^{\epsilon}}

\newcommand{\vectornorm}[1]{\|#1\|}
\newcommand{\vectornormbig}[1]{\Big\|#1\Big\|}
\newcommand{\vectornormmed}[1]{\big\|#1\big\|}

\newcommand{\class}{\textsc{Clash}\xspace}

\pagenumbering{arabic}
\pagestyle{normalheadings}
\setcounter{page}{1}


\begingroup
\long\def\symbolfootnote[#1]#2{\begingroup%
\def\thefootnote{\fnsymbol{footnote}}\footnote[#1]{#2}\endgroup}

\chapter{Linear Inverse Problems with \\Norm and Sparsity Constraints}
\label{chapter:cjk}

\begin{authors}
\AUname{Volkan Cevher}
\AUemail{volkan.cevher@epfl.ch}
\AUaffiliation{Laboratory for Information and Inference Systems}
\AUaddress{Ecole Polytechnique Federale de Lausanne}
\AUname{Sina Jafarpour}
\AUemail{sina2jp@yahoo-inc.com}
\AUaffiliation{Multimedia Research Group}
\AUaddress{Yahoo! Research}
\AUname{Anastasios Kyrillidis}
\AUemail{anastasios.kyrillidis@epfl.ch}
\AUaffiliation{Laboratory for Information and Inference Systems}
\AUaddress{Ecole Polytechnique Federale de Lausanne}\symbolfootnote[1]{Authors are in alphabetical order.}
\end{authors}

\begin{abstract}
We describe two nonconventional algorithms for linear regression, called GAME and CLASH. The salient characteristics of these approaches is that they exploit the convex $\ell_1$-ball and non-convex $\ell_0$-sparsity constraints jointly in sparse recovery. To establish the theoretical approximation guarantees of GAME and CLASH, we cover an interesting range of topics from game theory, convex and combinatorial optimization. We illustrate that these approaches lead to improved theoretical guarantees and empirical performance beyond convex and non-convex solvers alone.
\end{abstract} 
\section{Introduction}\label{sec: intro} 
\textit{Sparse approximation} is a fundamental problem in compressed sensing \cite{CRT1,donoho}, as well as in many other signal processing and machine learning applications including variable selection in regression \cite{lasso,lasso2,lasso3}, graphical model selection \cite{graphgame2,graphgame1}, and sparse principal component analysis \cite{spca1,spca2}. In sparse approximation, one is provided with a dimensionality reducing measurement matrix $\A \in \mathbb{R}^{\mm\times \n}$ ($\m< \n$), and a low dimensional vector $\f \in \R^\mm$ such that:
\begin{equation}\label{eq:obv}
  \obs = \A \xtrue + \noise,
\end{equation} where $ \xtrue \in \mathbb{R}^\n $ is the high-dimensional signal of interest and $ \mathbf{n} \in \mathbb{R}^{\mm} $ is a potential additive noise term with $ \vectornorm{\noise}_2 \leq \sigma $. 

In this work, we assume $ \xtrue $ is a $ k $-sparse signal or is sufficiently approximated by a $ k $-sparse vector. The goal of sparse approximation algorithms is then to find a sparse vector $\has \in \mathbb{R}^\n$ such that $\A\has-\f$ is small in an appropriate norm. In this setting, the $ \ell_0 $-minimization problem emerges naturally as a suitable solver to recover $ \xtrue $ in (\ref{eq:obv}):
\begin{equation}
	\begin{aligned}
	& \underset{\signal\in\mathbb{R}^\dimension}{\text{minimize}}
	& & \vectornorm{\signal}_0
	& \text{subject to}
	& & \vectornorm{\obs - \A \signal}_2 \leq \sigma,
	\end{aligned} \label{opt:00}
\end{equation} where $ \vectornorm{\signal}_0 $  counts the nonzero elements (the sparsity) of $ \signal $.

Unfortunately, solving (\ref{opt:00}) is a challenging task with exponential time complexity. Representing the set of all $ k $-sparse vectors as:
\begin{equation}\label{UoS}
\Delta_{\ell_0}(\sparsity)\doteq \{\as\in\R^\n: \|\as\|_0 \leq \sparsity \},
\end{equation} hard thresholding algorithms \cite{recipes,SP,cosamp,Blumensath_iterativehard,Foucart_hardthresholding} abandon this approach in favor of greedy selection where a putative $ \sparsity $-sparse solution is iteratively refined using local decision rules. To this end, hard thresholding methods consider the following $ \ell_0 $-constrained least squares problem formulation as an alternative to (\ref{opt:00}): \vspace{-0.1cm}
\begin{equation}
	\begin{aligned}
	& \underset{\as \in\mathbb{R}^\dimension}{\text{minimize}}
	& & \vectornorm{\obs - \A \as}_2^2
	& \text{subject to}
	& & \as \in \Delta_{\ell_0}(\sparsity).
	\end{aligned} \label{opt:01}
\end{equation}
These methods feature computational advantages and also are backed up with a great deal of theory for estimation guarantees.  

In contrast, convex optimization approaches change the problem formulations above by ``convexifying'' the combinatorial $ \ell_0 $-constraint with the sparsity inducing convex $ \ell_1 $-norm.\footnote{Note that this is not a true convexification, since the $ \ell_0 $-ball does not have a scale.} As a result, (\ref{opt:00}) is transformed into the $ \ell_1 $-minimization, also known as the Basis Pursuit (BP) problem \cite{BPDN}:
\begin{equation}
	\begin{aligned}
	& \underset{\as\in\mathbb{R}^\n}{\text{minimize}}
	& & \vectornorm{\as}_1
	& \text{subject to}
	& & \vectornorm{\obs - \A \as}_2 \leq \sigma.
	\end{aligned} \label{opt:02}
\end{equation} 
Similarly, the famous Lasso algorithm \cite{tibshirani96regression} can be considered as a relaxation of (\ref{opt:01}): \vspace{-0.1cm}
\begin{equation}
	\begin{aligned}
	& \underset{\as\in \mathbb{R}^\n}{\text{minimize}}
	& & \vectornorm{\obs - \A \as}_2^2
	& \text{subject to}
	& & \as \in \Delta_{\ell_1}(\tau),
	\end{aligned} \label{opt:03}
\end{equation} 
where $ \Delta_{\ell_1}(\tau) $ is the set of all vectors inside the hyper-diamond of radius $\tau$:
\begin{equation}\label{simplex}
\Delta_{\ell_1}(\tau)\doteq \{\as\in\R^\n: \|\as\|_1 \leq \tau \}.
\end{equation}

\begin{figure*}
\begin{minipage}[b]{0.42\linewidth}
\centering
{\subfigure[BP geometry]{\centerline{\resizebox{5cm}{!}{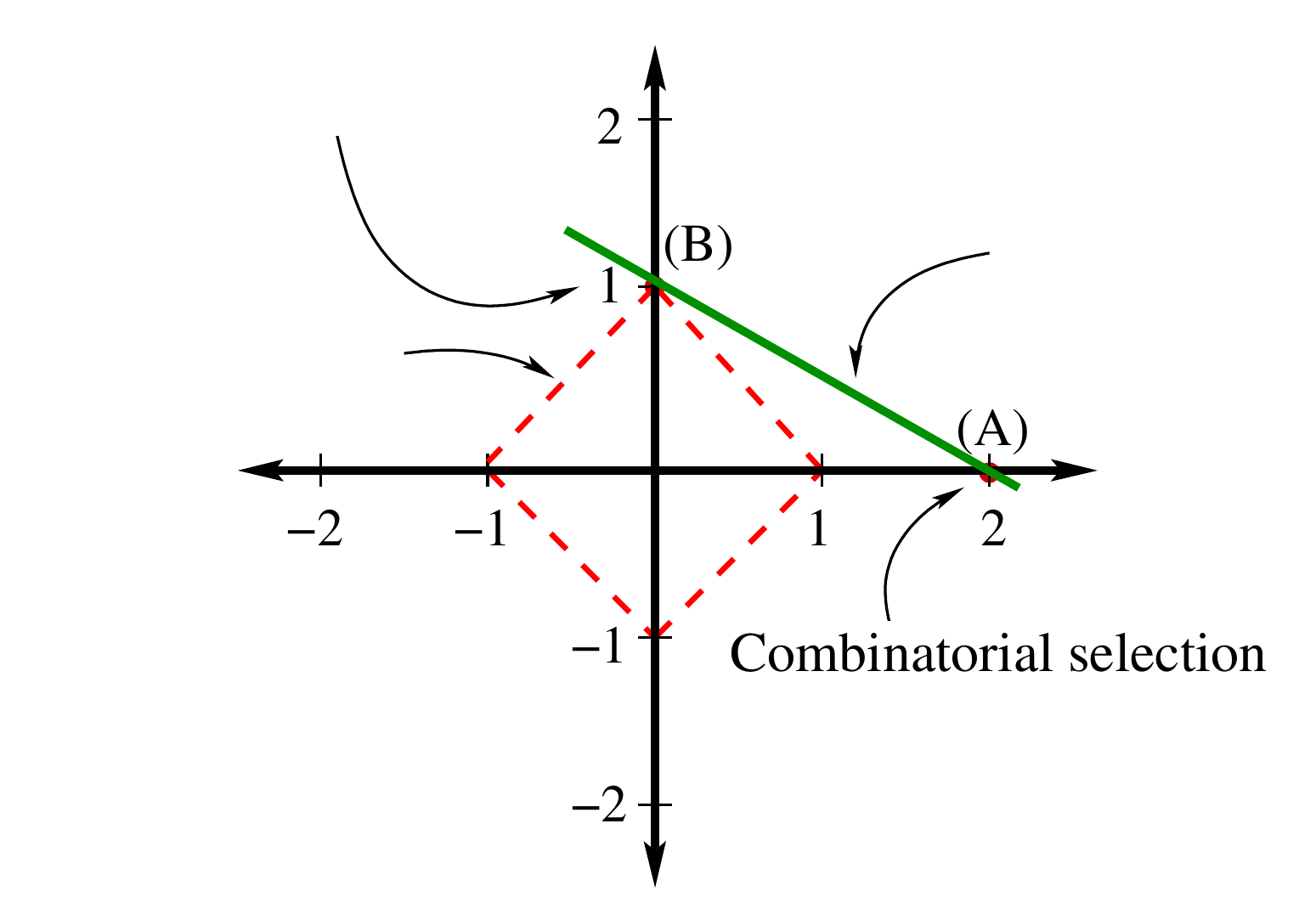}}}}
\end{minipage}
\hspace{0cm}
\begin{minipage}[b]{0.42\linewidth}
\centering
{\subfigure[Lasso geometry ]{\centerline{\resizebox{4.2cm}{!}{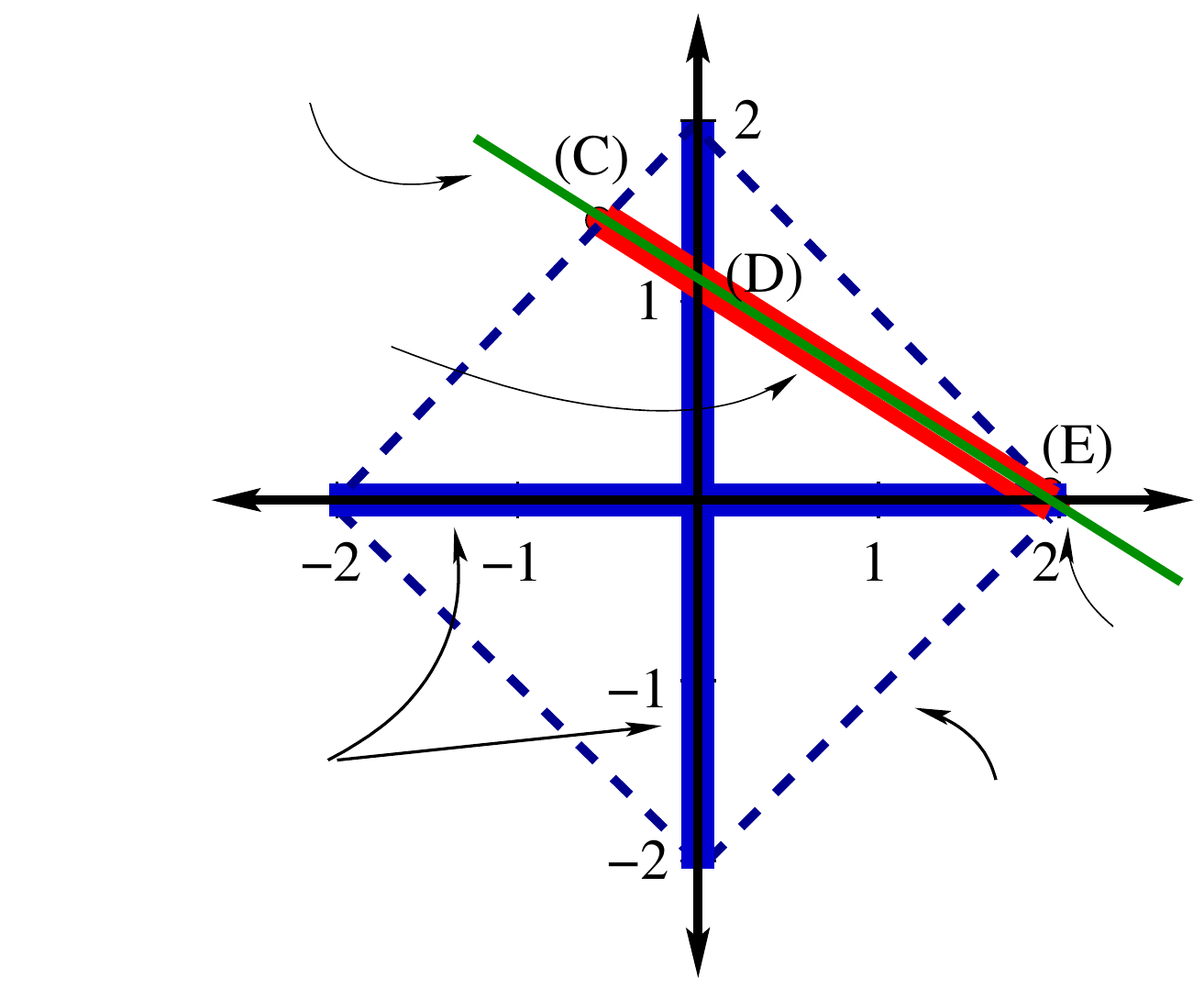}}}}
\end{minipage}
\caption{\small\sl Geometric interpretation of the selection process for a simple test case $ \obs = \A \bestsignal $ where $ \vectornorm{\bestsignal}_0 = 1 $. } \label{fig: grades}
\end{figure*}

While both convex and non-convex problem formulations can find the true problem solution under various theoretical assumptions, one can easily find examples in practice where either one can fail. Borrowing from \cite{normpursuits}, we provide an illustrative example in $ \mathbb{R}^2 $ for the noiseless case in Fig. 1.1. In (\ref{opt:00}), combinatorial-based approaches can identify the admissible set of 1-sparse solutions. If a greedy selection rule is used to arbitrate these solutions, then such an approach could pick (A). In contrast, the BP algorithm selects a solution (B), and misses the candidate solution (A) as it cannot exploit prior knowledge concerning the discrete structure of $ \bestsignal $.

To motivate our discussion in this book chapter, let us assume that we have the {\it true model parameters} $ \| \bestsignal \|_0 = \sparsity $ and $ \| \bestsignal \|_1 = \tau $. Let us then consider geometrically the---unfortunate but common---case where the kernel of $ \A $, $ {\rm ker}(\A) $, intersects with the tangent cone $ T_{\|\signal\|_1 \leq \tau}(\bestsignal) = \big \lbrace s(\mathbf{y} - \bestsignal): \|\mathbf{y}\|_1 \leq \tau \text{ and } s \geq 0 \big \rbrace $ at the true vector $ \bestsignal $ (cf., (E) in Fig. 1.1(b)). From the Lasso perspective, we are stuck with the large continuum of solutions based on the geometry, as described by the set $ \mathcal{I} = {\rm ker}(\A) \cap T_{\|\signal\|_1 \leq \tau}(\bestsignal) $, as illustrated in Figure 1.1(b) within the box.

Without further information about the discrete nature of $ \bestsignal $, a convex optimization algorithm solving the Lasso problem can arbitrarily select a vector from $ \mathcal{I} $. By forcing basic solutions in optimization, we can reduce the size of the solution space to $ \mathcal{L} = \mathcal{I} \cap \lbrace \|\signal \|_1 = 1\rbrace $, which is constituted by the sparse vectors (C) and (E). Note that $ \mathcal{L}$ might be still large in high dimensions. However, in this scenario, adding the $ \Delta_{\ell_0}(k) $ constraints, we can make precise selections (e.g., exactly 1-sparse), significantly reduce the candidate solution set, and, in many cases, can obtain the correct solution (E) if we leverage the norm constraint.

\textbf{Contents of this book chapter:} Within this context, we describe two efficient, sparse approximation algorithms, called {\it \textsc{GAME}} and {\it \textsc{Clash}}, that operate over sparsity and $\ell_1$-norm constraints. They address the following nonconvex problem:
\begin{equation}\label{spexact}
\begin{aligned}
	& \underset{\as\in\Delta_{\ell_0, \ell_1}(k,\tau)}{\text{minimize}}
	& & \|\A\as-\f\|_q	,
\end{aligned}
\end{equation}
where $\Delta_{\ell_1}(\tau)$ is the set of all $k$-sparse vectors in $ \Delta_{\ell_1}(\tau) $:
\begin{equation}\label{sparsesimplex}
 \Delta_{\ell_0, \ell_1}(k,\tau)\doteq\{\as\in\R^\n: \|\as\|_0\leq k\mbox{ and }\|\as\|_1 \leq \tau \}.
 \end{equation}

To introduce the {\it Game-theoretic Approximate Matching Estimator} ({\it \textsc{GAME}}) method, we reformulate (\ref{spexact}) as a zero-sum game. {\it \textsc{GAME}} then efficiently obtains a sparse approximation for the optimal game solution. {\it \textsc{GAME}} employs a primal-dual scheme, and require $\tilde{O}(k)$ iterations in order to find a $k$-sparse vector with $O\left(k^{-0.5}\right)$ additive approximation error.

To introduce the {\it Combinatorial selection and Least Absolute SHrinkage} operator {\it \textsc{Clash}}, we recall hard thresholding methods and explain how to incorporate the $\ell_1$ norm constraint. A key feature of the {\it \textsc{Clash}} approach is that it allows us to exploit ideas from the model-based compressive sensing (model-CS) approach, where selections can be driven by a structured sparsity model \cite{baraniuk2010model,clash}.

We emphasize again that since $\Delta_{\ell_0, \ell_1}(k,\tau)$ is not convex, the optimization problem \eqref{spexact} is not a convex optimization problem. However, we can still derive theoretical approximation guarantees of both algorithms. For instance, we can prove that for every dimension reducing matrix $\A$, and every measurement vector $\f$, {\it \textsc{{\it \textsc{GAME}}}} can find a vector $\has\in\Delta_{\ell_0, \ell_1}(k,\tau)$ with
\begin{equation}\label{sparseapp}
\|\A\has-f\|_q\leq\min_{\as\in\Delta_{\ell_0, \ell_1}(k,\tau)}\|\A\as-\f\|_q+\tilde{O}\left(\frac{1}{\sqrt{k}}\right),
\end{equation}
where $q$ is a positive integer. This sparse approximation framework surprisingly works for any matrix $\A$. Compared to the {\it \textsc{{\it \textsc{GAME}}}} algorithm, \class requires stronger assumptions on the measurement matrix for estimation guarantees. However, these assumptions, in the end, lead to improved empirical performance.

\section{Preliminaries}\label{sec: prelims}
Here, we cover basic mathematical background that is used in establishing algorithmic guarantees in the sequel.
\subsection{Bregman Projections}
\label{sec:bregmandef}
Bregman divergences or Bregman distances are an important family of distances that all share similar properties \cite{censor,bregman}.

\begin{definition}[Bregman Distance]
\label{bregmandist}
Let $\F:{\cal S}\rightarrow \R$ be a continuously-differentiable real-valued and strictly convex function defined on a closed convex set ${\cal S}$. The Bregman distance associated with $\F$ for points $\P$ and $\Q$ is:
$$\breg(\P,\Q)=\F(\P)-\F(\Q)-\inner{(\P-\Q)}{\nabla \F(\Q)}.$$
\end{definition}

\begin{figure}[t]
\centering
   \includegraphics[scale=0.25] {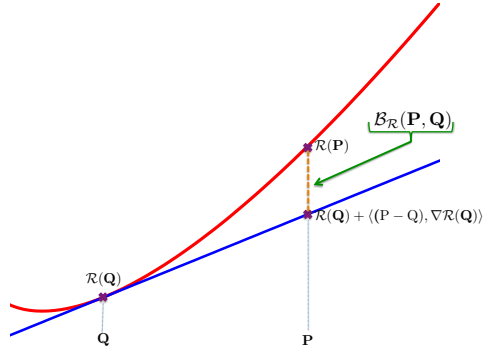}
\caption{The Bregman divergence associated with a continuously-differentiable real-valued and strictly convex function $\F$ is the vertical distance at $\P$ between the graph of $\F$ and the line tangent to the graph of $\F$ in $\Q$.}
  \label{fig_breg_1}
\end{figure}

Table~\ref{tab_breg_1} summarizes examples of the most widely used Bregman functions and the corresponding Bregman distances.

\begin{table}[ht]
\caption{Summary of the most popular Bregman functions and their corresponding Bregman distances. Here $\A$ is a positive semidefinite matrix.}
\begin{center}
\begin{tabular}{|c|c|c|}
\hline
  \multirow{2}{*}{Name} & Bregman & Bregman
 \\ & Function $(\F(\P))$ & Distance $(\breg(\P,\Q))$
 \\\hline
 \hline
  Squared &    \multirow{2}{*}{$\|\P\|_2^2$} & \multirow{2}{*}{$\|\P-\Q\|_2^2$}\\
  Euclidean & &\\
  \hline
  Squared &    \multirow{2}{*}{$\inner{\P}{\bs{\Phi}\P}$} & \multirow{2}{*}{$\inner{(\P-\Q)}{\bs{\Phi}(\P-\Q)}$}\\
   Mahalanobis & &\\
  \hline
 \multirow{2}{*}{Entropy} &\multirow{2}{*}{$\sum_i {\rm P}_i \log {\rm P}_i-{\rm P}_i$} & \multirow{2}{*}{$\sum_i {\rm P}_i \log \frac{{\rm P}_i}{{\rm Q}_i} -\sum_i ({\rm P}_i-{\rm Q}_i)$}\\ & &
 \\\hline
  \multirow{2}{*}{Itakura-Saito} &\multirow{2}{*}{$\sum_i -\log {\rm P}_i$} & \multirow{2}{*}{$\sum_i \left(\frac{{\rm P}_i}{{\rm Q}_i}- \log \frac{{\rm P}_i}{{\rm Q}_i}+1\right)$} \\ & &
 \\\hline
\end{tabular}
\end{center}
\label{tab_breg_1}
\end{table}

The Bregman distance has several important properties that we will use later in analyzing our sparse approximation algorithm.
\begin{theorem}
\label{bregprop}
Bregman distance satisfies the following properties:
\begin{itemize}
\item (P1). $\breg(\P,\Q)\geq 0$, and the equality holds if and only if $\P=\Q$.
\item (P2). For every fixed $\Q$ if we define ${\cal G}(\P)=\breg(\P,\Q)$, then
$$\nabla {\cal G}(\P)=\nabla \F(\P)-\nabla \F(\Q).$$
\item (P3). Three point property: For every $\P,\Q$ and $\bs{\rm T}$ in ${\cal S}$
\begin{align}\nonumber\breg(\P,\Q)&=\breg(\P,\bs{\rm T})+\breg (\bs{\rm T},\Q) +\inner{(\P-\bs{\rm T})}{\nabla \F(\Q)-\nabla \F(\bs{\rm T})}.\end{align}
\item (P4). For every $\P,\Q\in{\cal S}$, $$\breg(\P,\Q)+\breg(\Q,\P)=\inner{(\P-\Q)}{(\nabla \F(\P)-\nabla \F(\Q))}.$$
\end{itemize}
\end{theorem}
\begin{proof}
All four properties follow directly from Definition~\ref{bregmandist}.
\end{proof}

Now that we are equipped with the properties of Bregman distances, we are ready to define Bregman projections of points into convex sets.
\begin{definition}[Bregman Projection]
\label{bregmanproj}
Let $\F:{\cal S}\rightarrow \R$ be a continuously-differentiable real-valued and strictly convex function defined on a closed convex set ${\cal S}$. Let $\Omega$ be a closed subset of ${\cal S}$. Then, for every point $\Q$ in ${\cal S}$, the Bregman projection of $\Q$ into $\Omega$, denoted as $\proj(\Q)$ is
$$\proj(\Q)\doteq \arg\min_{\P\in\Omega} \breg(\P,\Q).$$
\end{definition}

Bregman projections satisfy a generalized Pythagorean Theorem.

\begin{theorem}[Generalized Pythagorean Theorem \cite{censor}]
\label{pythag}
Let $\F:{\cal S}\rightarrow \R$ be a continuously-differentiable real-valued and strictly convex function defined on a closed convex set ${\cal S}$. Let $\Omega$ be a closed subset of ${\cal S}$. Then for every $\P\in \Omega$ and $\Q\in{\cal S}$
\begin{equation}
\breg(\P,\Q)\geq \breg(\P, \proj(\Q)) + \breg(\proj(\Q),\Q),
\end{equation}
and in particular
\begin{equation}
\breg(\P,\Q)\geq \breg(\P, \proj(\Q)).
\end{equation}
\end{theorem}

 We refer the reader to \cite{censor}, or \cite{bregmanbook} for a proof of this theorem and further discussions.

\subsection{Euclidean Projections onto the $\ell_0$ and the $\ell_1$-ball}
Here, we describe two of key actors in sparse approximation.

\textbf{Projections onto combinatorial sets:}
The Euclidean projection of a signal $ \mathbf{w} \in \mathbb{R}^\n $ on the subspace defined by $ \Delta_{\ell_0}(\sparsity) $ is provided by:
\begin{align}{\label{eq:proj}}
\mathcal{P}_{\Delta_{\ell_0}(\sparsity)}(\mathbf{w}) = \argmin_{\as: \as \in \Delta_{\ell_0}(\sparsity)} \vectornorm{\as - \bf{w}}_2,
\end{align} whose solution is hard thresholding. That is, we sort the coefficients of $\mathbf{w}$ in decreasing magnitude and keep the top $\sparsity$ and threshold the rest away. This operation can be done in $ O(n \log n) $ time complexity via simple sorting routines. 

\textbf{Projections onto convex norms:}
Given $ \mathbf{w} \in \mathbb{R}^\n $, the Euclidean projection onto a convex $ \ell_1 $-norm ball of radius at most $ \tau $ defines the optimization problem:
\begin{align}
\mathcal{P}_{\Delta_{\ell_1}(\tau)}(\mathbf{w}) = \argmin_{\as: \as \in \Delta_{\ell_1}(\tau)} \vectornorm{\as - \bf{w}}_2,
\end{align}
whose solution is soft thresholding. That is, we decrease the magnitude of all the coefficients by a constant value just enough to meet the $\ell_1$ norm constraint. A solution can be obtained in $ O(n \log n) $ time complexity with simple sorting routines, similar to above.

\subsection{Restricted Isometry Property}
In order to establish stronger theoretical guarantees for the algorithms, it is necessary to use Restricted Isometry Property (RIP) assumption. For each positive integers $q$ and $k$, and each $\epsilon$ in $(0,1)$, an $\m\times \n$ matrix $\A$ satisfies the $(k, \epsilon)$ RIP in $\ell_q$ norm ($(k, \epsilon)$ RIP-$q$) \cite{berinde2008combining,sinathesis}, if for every $k$-sparse vector $\as$,
$$(1-\epsilon)\|\as\|_q\leq \|\A\as\|_q\leq (1+\epsilon)\|\as\|_q.$$
This assumption implies near isometric embedding of the sparse vectors by the matrix $\A$. We just briefly mention that such matrices can be constructed \emph{randomly} using certain classes of distributions \cite{sinathesis}.
 
\section{The GAME Algorithm}
\label{sec:sparsegame}
\subsection{A Game Theoretic Reformulation of Sparse Approximation}
We start by defining a zero-sum game and then proving that the sparse approximation problem of Equation~\eqref{spexact} can be reformulated as a zero-sum game.

\begin{definition}[Zero-sum games \cite{gamebook}]
Let ${\cal A}$ and ${\cal B}$ be two closed sets. Let ${\cal L}:{\cal A}\times {\cal B}\rightarrow \R$ be a function. The value of a zero sum game, with domains ${\cal A}$ and ${\cal B}$ with respect to a function ${\cal L}$ is defined as
\begin{equation}
\min_{\bs{a}\in {\cal A}}\max_{\bs{b}\in{\cal B}} {\cal L}(\bs{a},\bs{b}).
\end{equation}
\end{definition}

The function $\Lo$ is usually called the \textit{loss function}. A zero-sum game can be viewed as a game between two players Mindy and Max in the following way. First, Mindy finds a vector $\bs{a}$, and then Max finds a vector ${\bs{b}}$. The loss that Mindy suffers\footnote{which is equal to the gain that Max obtains as the game is zero-sum.} is ${\cal L}(\bs{a},\bs{b})$. The game-value of a zero-sum game is then the loss that Mindy suffers if both Mindy and Max play with their optimal strategies.

Von Neumann's well-known Minimax Theorem \cite{minimax1,minimax2} states that if both ${\cal A}$ and ${\cal B}$ are convex compact sets, and if the loss function ${\cal L}(\bs{a},\bs{b})$ is convex with respect to $\bs{a}$, and concave with respect to $\bs{b}$, then the game-value is independent of the ordering of the game players.

\begin{theorem}[Von Neumann's Minimax Theorem \cite{minimax1}]
\label{von}
Let ${\cal A}$ and ${\cal B}$ be closed convex sets, and let ${\cal L}:{\cal A}\times {\cal B}\rightarrow \R$ be a function which is convex with respect to its first argument, and concave with respect to its second argument. Then $$\inf_{\bs{a}\in{\cal A}} \sup_{\bs{b}\in{\cal B}} {\cal L}(\bs{a},\bs{b})= \sup_{\bs{b}\in{\cal B}}\inf_{\bs{a}\in{\cal A}} {\cal L}(\bs{a},\bs{b}).$$
\end{theorem}

For the history of the Minimax Theorem see \cite{minimaxrob}. The Minimax Theorem tells us that for a large class of functions ${\cal L}$, the values of the min-max game in which Mindy goes first is identical to the value of the max-min game in which Max starts the game. The proof of the Minimax Theorem is provided in \cite{gameproof}.

Having defined a zero-sum game, and the Von Neumann Minimax Theorem, we next show how the sparse approximation problem of Equation~\eqref{spexact} can be reformulated as a zero-sum game. Let $p\doteq \frac{q}{q-1}$, and define \begin{equation}\label{simpp}\simp_p\doteq \{\P\in\R^\m: \|\P\|_p\leq 1\} .\end{equation} Define the loss function ${\cal L}: \simp_p\times \Delta_{\ell_1}(\tau) \rightarrow \R$ as \begin{equation}\label{lossdef}{\cal L}(\P,\as)\doteq \inner{\P}{(\A\as-f)}.\end{equation} Observe that the loss-function is bilinear. Now it follows from H\"{o}lder inequality that for every $\as$ in $\Delta_{\ell_0, \ell_1}(k,\tau)$, and for every $\P$ in $\simp_p$
\begin{equation}
\label{firstminmax}
{\cal L}(\P,\as) =\inner{\P}{ (\A\as-\f)} \leq \|\P\|_p \|\A\as-\f\|_q \leq \|\A\as-\f\|_q.
\end{equation}
The inequality of Equation~\eqref{firstminmax} becomes equality for $$P^*_i= \frac{(\A\as-\f)_i^{q/p}}{\left(\sum_{i=1}^\m (\A\as-\f )_i^{q} \right)^{1/p}}.$$ Therefore
\begin{equation}\label{secondminmax}
\max_{\P\in\simp_p}{\cal L}(\P,\as) =\max_{\P\in\simp_p}\inner{\P}{(\A\as-\f)}=\inner{\P^*}{(\A\as-\f)}= \|\A\as-\f\|_q.
\end{equation}
Equation~\eqref{secondminmax} is true for every $\as\in\Delta_{\ell_1}(\tau)$. As a result, by taking the minimum over $\Delta_{\ell_0, \ell_1}(k,\tau)$ we get
$$\min_{\as\in\Delta_{\ell_0, \ell_1}(k,\tau)}\|\A\as-\f\|_q=\min_{\as\in\Delta_{\ell_0, \ell_1}(k,\tau)} \max_{\P\in\simp_p}{\cal L}(\P,\as).$$
Similarly by taking the minimum over $\Delta_{\ell_1}(\tau)$ we get
\begin{equation}\label{newdemand}\min_{\as\in\Delta_{\ell_1}(\tau)}\|\A\as-\f\|_q=\min_{\as\in\Delta_{\ell_1}(\tau)} \max_{\P\in\simp_p}{\cal L}(\P,\as).\end{equation}
Solving the sparse approximation problem of Equation~\eqref{spexact} is therefore equivalent to finding the optimal strategies of the game \begin{equation}\label{temp1}\min_{\as\in\Delta_{\ell_0,\ell_1}(k,\tau)} \max_{\P\in\simp_p}{\cal L}(\P,\as).\end{equation} In the next section we provide a primal-dual algorithm that approximately solves this min-max game. Observe that since $\Delta_{\ell_0,\ell_1}(k,\tau)$ is a subset of $\Delta_{\ell_1}(\tau)$, we always have
$$\min_{\as\in\Delta_{\ell_1}(\tau)} \max_{\P\in\simp_p}{\cal L}(\P,\as)\leq \min_{\as\in\Delta_{\ell_0,\ell_1}(k,\tau)} \max_{\P\in\simp_p}{\cal L}(\P,\as),$$ and therefore, in order to approximately solve the game of Equation~\eqref{temp1}, it is sufficient to find $\has\in \Delta_{\ell_0,\ell_1}(k,\tau)$ with
\begin{equation}\label{last_game} \max_{\P\in\simp_p}{\cal L}(\P,\has)\approx \min_{\as\in\Delta_{\ell_1}(\tau)} \max_{\P\in\simp_p}{\cal L}(\P,\as).\end{equation}

\subsection{Algorithm Description}
\label{sec:game_game}
In this section we provide an efficient algorithm for approximately solving the problem of sparse approximation in $\ell_q$ norm, defined by Equation~\eqref{sparseapp}. Let ${\cal L}(\P,\as)$ be the loss function defined by Equation~\eqref{lossdef}, and recall that in order to approximately solve Equation~\eqref{sparseapp}, it is sufficient to find a sparse vector $\has\in \Delta_{\ell_0,\ell_1}(k,\tau)$ such that
\begin{equation}\label{newgame}\max_{\P \in\simp_p} {\cal L}(\P,\has)\approx  \min_{\as'\in \Delta_{\ell_1}(\tau)}\max_{\P \in\simp_p} {\cal L}(\P,\as).\end{equation}

The original sparse approximation problem of Equation~\eqref{sparseapp} is NP-complete, but it is computationally feasible to compute the value of the min-max game
\begin{equation}\label{newgame2} \min_{\as'\in \Delta_{\ell_1}(\tau)}\max_{\P \in\simp_p} {\cal L}(\P,\as).\end{equation}
The reason is that the loss function ${\cal L}(\P,\as)$ of Equation~\eqref{lossdef} is a bilinear function, and the sets $\Delta_{\ell_1}(\tau)$, and $\simp_p$ are both convex and closed.

Therefore, finding the game values and optimal strategies of the game of Equation~\eqref{newgame2} is equivalent to solving a convex optimization problem and can be done using off-the-shelf non-smooth convex optimization methods \cite{nonsmooth,nonsmooth2}. However, if an off-the-shelf convex optimization method is used, then there is no guarantee that the recovered strategy $\has$ is also sparse. We need an approximation algorithm that finds near-optimal strategies $\has$ and $\hat{\P}$ for Mindy and Max with the additional guarantee that Mindy's near optimal strategy $\has$ is sparse.

Here we introduce the {\em Game-theoretic Approximate Matching Estimator} (GAME) algorithm which finds a sparse \textit{approximation} to the min-max optimal solution of the game defined in Equation~\eqref{newgame2}. The \muse algorithm relies on the general primal-dual approach which was originally applied to developing strategies for repeated games  \cite{gameproof} (see also \cite{hazan} and \cite{hazan2}). The  pseudocode of the \muse Algorithm is provided in Algorithm~\ref{alggame}.

\begin{algorithm}[t]
\caption{GAME Algorithm for Sparse Approximation in $\ell_q$-norm.}
\textbf{Inputs:} $\m$-dimensional vector $\f$, $\m\times \n$ matrix $\A$, number of iterations $T$, sparse approximation norm $q$, Bregman function $\F$ and regularization parameter $\eta$.\\
\textbf{Output:} $\n$-dimensional vector $\has$
\label{alggame}
\end{algorithm}

The GAME Algorithm can be viewed as a repeated game between two players Mindy and Max who iteratively update their current strategies $\P^t$ and $\as^t$, with the aim of ultimately finding near-optimal strategies based on a $T$-round interaction with each other. Here, we briefly explain how each player updates his/her current strategy based on the new update from the other player.

Recall that the ultimate goal is to find the solution of the game $$\min_{\as'\in \Delta_{\ell_1}(\tau)}\max_{\P \in\simp_p} {\cal L}(\P,\as).$$ At the begining of each iteration $t$, Mindy receives the updated value $\P^{t}$ from Max. A greedy Mindy only focuses on Max's current strategy, and updates her current strategy to $\as^t= \arg\min_{\as\in\Delta_{\ell_1}(\tau)}{\cal L}(\P^{t},\as).$ In the following lemma we show that this is indeed what our Mindy does in the first three steps of the main loop.

\begin{lemma}\label{imp_game}
Let $\P^{t}$ denote Max's strategy at the begining of iteration $t$. Let $\bs{r}^t=\A^\top \P^{t}$, and let $i$ denote the index of a largest (in magnitude) element of $\bs{r}^t$. Let $\as^t$ be a $1$-sparse vector with ${\rm Supp}(\as^t)=\{i\}$ and with $\alpha_i^t=-\tau\,{\rm Sign}\left(r_{i}^t\right)$. Then $\as^t= \arg\min_{\as\in\Delta_{\ell_1}(\tau)}{\cal L}(\P^{t},\as).$
\end{lemma}
\begin{proof}
Let $\tilde{\as}$ be any solution $\tilde{\as}= \arg\min_{\as\in\Delta_{\ell_1}(\tau)}{\cal L}(\P^{t},\as)$. It follows from the bilinearity of the loss function (Equation~\eqref{lossdef}) that
\begin{align}\nonumber\tilde{\as}&= \arg\min_{\as\in\Delta_{\ell_1}(\tau)}{\cal L}(\P^{t},\as)\\\nonumber&=\arg\min_{\as\in\Delta_{\ell_1}(\tau)} \inner{\P^{t}}{\A\as-\f}= \arg\min_{\as\in\Delta_{\ell_1}(\tau)} \inner{\A^\top \P^{t}}{\as}.\end{align}
Hence, H\"{o}lder inequality yields that for every $\as^\#\in\Delta_{\ell_1}(\tau)$,
\begin{equation}\label{holdeq}\inner{\A^\top \P^{t}}{\as^\#}\geq -\|\as^\#\|_1 \|\A^\top \P^{t} \|_\infty \geq -\tau  \|\A^\top \P^{t}\|_\infty.\end{equation}
Now let $\as^t$ be a $1$-sparse vector with ${\rm Supp}(\as^t)=\{i\}$ and $\as_i^t=-\tau\,{\rm Sign}\left(r_{i}^t\right)$. Then $\as^t\in \Delta_{\ell_1}(\tau)$, and
$$\inner{\A^\top \P^{t}}{\as^t}=-\tau \|\A^\top\P^{t}\|_\infty.$$
In other words, for $\as^t$ the Holder inequality is an equality. Hence $\as^t$ is a minimizer of $\inner{\A^\top \P^{t}}{\as}$.
\end{proof}

Thus far we have seen that at each iteration Mindy always finds a $1$-sparse solution $\as^t=\arg\min_{\as\in\Delta_{\ell_1}(\tau)}{\cal L}(\P^{t},\as)$. Mindy then sends her updated strategy $\as^t$ to Max, and now it is Max's turn to update his strategy. A greedy Max would prefer to update his strategy as $\P^{t+1}=\arg\max_{\P\in\simp_p} {\cal L}(\P,\as^t).$ However, our Max is more conservative and prefers to stay close to his previous value $\P^{t}$. In other words, Max has two competing objectives
\begin{enumerate}
\item Maximizing ${\cal L}(\P,\as^t)$, or equivalently minimizing $-{\cal L}(\P,\as^t)$.
\item Remaining close to the previous strategy $\P^t$, by minimizing $\breg(\P,\P^{t-1})$.
\end{enumerate}
Let $${\cal L}_{\F}(\P)\doteq -\eta{\cal L}(\P,\as^t) + \breg(\P,\P^{t}),$$ be a regularized loss function which is a linear combination of the two objectives above.

A conservative Max then tries to minimize a combination of the two objectives above by minimizing the regularized loss function \begin{equation}\label{hardy}\P^{t+1}=\arg\min_{\P\in\simp_p} {\cal L}_{\F}(\P)=\arg\min_{\P\in\simp_p}-\eta{\cal L}(\P,\as^t) + \breg(\P,\P^{t}).\end{equation}

Unfortunately, it is not so easy to efficiently solve the optimization problem of Equation~\eqref{hardy} at every iteration. To overcome this difficulty, our Max first ignores the constraint $\P^{t+1}\in\simp_p$, and instead finds a global optimizer of ${\cal L}_{\F}(\P)$ by setting $\nabla {\cal L}_{\F}(\P)={\bs 0}_\m$, and then projects back the result to $\simp_p$ via a Bregman projection.

More precisely, it follows from the Property~(P2) of Bregman distance (Theorem~\ref{bregprop}) that for every $\P$
$$\nabla {\cal L}_{\F}(\P)= -\eta(\A\as^t-\f) + \nabla \F(\P)-\nabla \F(\P^{t}),$$
and therefore if $\Q^t$ is a point with $$\nabla \F(\Q^t)= \nabla \F(\P^{t-1})+\eta(\A\as^t-\f),$$ then $\nabla {\cal L}_{\F}(\Q^t)=\bs{0}_\m$.

The vector $\Q^t$ is finally projected back to $\simp_p$ via a Bregman projection to ensure that Max's new strategy is in the feasible set $\simp_p$.

\subsection{The GAME Guarantees}
\label{sec:analgame}
In this section we prove that the \muse algorithm finds a near-optimal solution for the sparse approximation problem of Equation~\eqref{sparseapp}. The analysis of the \muse algorithm relies heavily on the analysis of the generic primal-dual approach. This approach originates from the \textit{link-function methodology} in computational optimization \cite{hazan2,warmuth}, and is related to the \textit{mirror descent} approach in the optimization community \cite{mirror1,mirror2} . The primal-dual Bregman optimization approach is widely used in online optimization applications including portfolio selection \cite{portfolio1,portfolio2}, online learning \cite{online}, and boosting \cite{bregman2,bregman3}.

However, there is a major difference between the sparse approximation problem and the problem of online convex optimization. In the sparse approximation problem, the set ${\cal A}=\Delta_{\ell_0, \ell_1}(k,\tau)$ is not convex anymore; therefore, there is no guarantee that an online convex optimization algorithm outputs a sparse strategy $\has$. Hence, it is not possible to directly translate the bounds from the online convex optimization scheme to the sparse approximation scheme.

Moreover, as discussed in Lemma~\ref{imp_game}, there is also a major difference between the Mindy players of the \muse algorithm and the general Mindy of general online convex optimization games. In the \muse algorithm, Mindy is not a \textit{blackbox adversary} that responds with an update to her strategy based on Max's update. Here, Mindy always performs a greedy update and finds the \textit{best} strategy as a response to Max's update. Moreover, our Mindy always finds a \textit{$1$-sparse} {new strategy}.
That is, she looks among all best responses to Max's update, and finds a $1$-sparse strategy among them.

As we will see next, the combination of cooperativeness by Mindy, and standard ideas for bounding the regret in online convex optimization schemes, enables us to analyze the \muse algorithm for sparse approximation. The following lemma bounds the regret loss of the primal-dual strategy in online convex optimization problems and is proved in \cite{hazan}.

\begin{theorem}
\label{hazanregret}
Let $q$ and $T$ be positive integers, and let $p=\frac{q}{q-1}$. Suppose that $\F$ is such that for every $\P,\Q\in\simp_p$, $\breg(\P,\Q)\geq \|\P-\Q\|_p^2$, and let
\begin{equation}
\label{Gbound}
{\rm G}=\max_{\as\in\Delta_{\ell_0, \ell_1}(1,\tau)}\|\A\as-\f\|_q.
\end{equation}
Also assume that for every $\P\in\simp_p$, we have $\breg(\P,\P^1) \leq {\rm D}^2$. Suppose $$\langle(\P^1,\as^1),\cdots,(\P^T,\as^T)\rangle$$ is the sequence of pairs generated by the \muse Algorithm after $T$ iterations with $\eta=\frac{2{\rm D}}{{\rm G}\sqrt{T}}$. Then
   \begin{align}
\nonumber \max_{\P\in\simp_p}  \frac{1}{T} \sum_{t=1}^T {\cal L}(\P,\as^t)  \leq  \frac{1}{T}\sum_{t=1}^T {\cal L}(\P^t,\as^t)+\frac{{\rm DG}}{2\sqrt{T}}.
   \end{align}
\end{theorem}
\begin{proof}
The proof of Theorem~\ref{hazanregret} is based on the geometric properties of the Bregman functions, and is provided in \cite{hazan}.
\end{proof}

Next we use Theorem~\ref{hazanregret} to show that the \muse algorithm after $T$ iterations finds a $T$-sparse vector $\has$ with near-optimal value $\|\A\has-\f\|_q.$

            \begin{theorem}\label{FSt_dant}
 Let $q$ and $T$ be positive integers, and let $p=\frac{q}{q-1}$. Suppose that for every $\P,\Q\in\simp_p$, the function $\F$ satisfies $\breg(\P,\Q)\geq \|\P-\Q\|_p^2$, and let
\begin{equation}
{\rm G}=\max_{\as\in\Delta_{\ell_0, \ell_1}(1,\tau)}\|\A\as-\f\|_q.
\end{equation}
Also assume that for every $\P\in\simp_p$, we have $\breg(\P,\P^1) \leq {\rm D}^2$. Suppose $$\langle(\P^1,\as^1),\cdots,(\P^T,\as^T)\rangle$$ is the sequence of pairs generated by the \muse Algorithm after $T$ iterations with $\eta=\frac{2{\rm D}}{{\rm G}\sqrt{T}}$. Let $\has=\frac{1}{T}\sum_{t=1}^T \as^t$ be the output of the \muse algorithm. Then $\has$ is a $T$-sparse vector with $\|\has\|_1\leq \tau$ and
    \begin{equation}\label{final}\|\A\has-\f\|_q\leq \min_{\as\in\Delta_{\ell_0, \ell_1}(T,\tau)}\|\A\as-\f\|_q +  \frac{{\rm DG}}{2\sqrt{T}}.\end{equation}
    \end{theorem}
    \begin{proof}
    It follows from Step~2. of Algorithm~\ref{alggame} that every $\as^t$ is $1$-sparse and $\|\as^t\|_1=\tau.$ Therefore, $\has=\frac{1}{T}\sum_{t=1}^T\as^t$ can have at most $T$ non-zero entries and moreover $\|\has\|_1\leq \frac{1}{T} \sum_{t=1}^T \|\as^t\|_1 \leq \tau$. Therefore $\has$ is in $\Delta_{\ell_0, \ell_1}(T,\tau)$.

Next we show that the Equation~\ref{final} holds for $\has$. Let $\hat{\P}=\frac{1}{T}\sum_{t=1}^T\P^t.$ Observe that
       \begin{align}
   \min_{\as\in\Delta_{\ell_1}(\tau)}\max_{\P\in\simp_p} \Lo\left(\P,\as\right) &\stackrel{(e)}{=} \max_{\P\in\simp_p} \min_{\as\in\Delta_{\ell_1}(\tau)} \Lo\left(\P,\as\right) \nonumber \\ &\stackrel{(f)}{\geq} \min_{\as\in\Delta_{\ell_1}(\tau)} \Lo\left(\hat{\P},\as\right) \nonumber
\\\nonumber &\stackrel{(g)}{\geq} \frac{1}{T}\min_{\as\in \Delta_{\ell_1}(\tau)} \sum_{t=1}^T \Lo(\P^t,\as)
\\\nonumber &\stackrel{(h)}{\geq} \frac{1}{T}\sum_{t=1}^T \min_{\as\in\Delta_{\ell_1}(\tau)} \Lo(\P^t,\as) =^i  \frac{1}{T}\sum_{t=1}^T \Lo(\P^t,\as^t)
  \\\nonumber &\stackrel{(j)}{\geq}  \max_{\P\in\simp_p} \Lo\left(\P,\frac{1}{T}\sum_{t=1}^T \as^t\right)- \frac{{\rm DG}}{2\sqrt{T}}.
    \end{align}
      Equality (e) is the minimax Theorem (Theorem~\ref{von}). Inequality (f) follows from the definition of the $\max$ function. Inequalities (g) and (h) are consequences of the bilinearity of $\Lo$ and concavity of the $\min$ function. Equality (i) is valid by the definition of $\as^t$, and Inequality (j) follows from Theorem~\ref{hazanregret}. As a result
        \begin{align}
     \label{main}
     \|\A\has-\f\|_q&= \max_{\P\in\simp_p} \Lo\left(\P,\has\right) \leq  \min_{\as\in\Delta_{\ell_1}(\tau)} \max_{\P\in\simp_p} \Lo(\P,\as) \nonumber \\\nonumber& +  \frac{{\rm DG}}{2\sqrt{T}}= \min_{\as\in\Delta_{\ell_1}(\tau)} \|\A\as-\f\|_q +  \frac{{\rm DG}}{2\sqrt{T}} .
     \end{align}
\end{proof}

    \begin{remark}
    In general, different choices for the Bregman function may lead to different convergence bounds with different running times to perform the new projections and updates. For instance, a multiplicative update version of the algorithm can be derived by using the Bregman divergence based on the Kullback-Leibler function, and an additive update version of the algorithm can be derived by using the Bregman divergence based on the squared Euclidean function.
     \end{remark}
     
     Theorem~\ref{FSt_dant} is applicable to \textit{any} sensing matrix. Nevertheless, it does not guarantee that the estimate vector $\has$ is close enough to the target vector $\sas$. However, if the sensing matrix satisfies the RIP-$q$ property, then it is possible to bound the data-domain error $\|\has-\sas\|_q$ as well.
     
       \begin{theorem}\label{FSt_data}
 Let $q$ $k$, and $T$ be positive integers, let $\epsilon$ be a number in $(0,1)$, and let $p=\frac{q}{q-1}$. Suppose that for every $\P,\Q\in\simp_p$, the function $\F$ satisfies $\breg(\P,\Q)\geq \|\P-\Q\|_p^2$, and let $\A$ be an $\m\times\n$ sensing matrix satisfying the $(k+T,\epsilon)$ RIP-$q$ property. Let $\sas$ be a $k$-sparse vector with $\|\sas\|_1\leq \tau$, let $\e$ be an arbitrary noise vector in $\mathbb{R}^\m$, and set $\f=\A\sas+\e$. Let ${\rm G}$, ${\rm D}$, and $\eta$ be as of Theorem~\ref{FSt_dant}, and let let $\has$ be the output of the \muse algorithm after $T$ iterations. Then $\has$ is a $T$-sparse vector with $\|\has\|_1\leq \tau$ and
    \begin{equation}\label{final_data}
    \|\has-\sas\|_q\leq  \frac{2\|\e\|_q++ \frac{{\rm DG}}{2\sqrt{T}}}{(1-\epsilon)}.\end{equation}
    \end{theorem}
    \begin{proof}
   Since $\has$ is $T$-sparse and $\sas$ is $k$-sparse, $\has-\sas$ is $(T+k)$-sparse. Therefore, it follows from the RIP-$q$ property of the sensing matrix that 
   \begin{align}
 & (1-\epsilon)\|\has-\sas\|_q \leq \|\A(\has-\sas)\|_q \leq \|\A\has-\f\|_q +\|\e\|_q \\\nonumber &\leq 
\|\A\sas-\f\|_q +  \frac{{\rm DG}}{2\sqrt{T}}+\|\e\|_q= 2\|\e\|_q+ \frac{{\rm DG}}{2\sqrt{T}}.
   \end{align}
    \end{proof}

\section{The CLASH Algorithm}\label{sec: CLASH}

\subsection{Hard Thresholding Formulations of Sparse Approximation}

As already stated, solving (\ref{opt:00}) is NP-hard and exhaustive search over $ \binom{\n}{k} $ possible support set configurations of the $ k $-sparse solution is mandatory. Contrary to this brute-force approach, hard thresholding algorithms \cite{recipes,SP,cosamp,Blumensath_iterativehard,Foucart_hardthresholding} navigate through the low-dimensional $ k $-sparse subspaces, {\it pursuing} an appropriate support set such to minimize the data error in (\ref{opt:01}). To achieve this, these approaches apply greedy support set selection rules to iteratively compute and refine a putative solution $ \as_{i} $ using only first-order information $ \nabla f(\as_{i-1}) $ at each iteration $ i $.

Subspace Pursuit (SP) \cite{SP} algorithm is a combinatorial greedy algorithm that borrows both from Orthogonal Matching Pursuit (OMP) and Iterative Hard Thresholding \cite{Blumensath_iterativehard} (IHT) methods. A sketch of the algorithm is given in Algorithm 2. The basic idea behind SP consists in looking for a good support set by iteratively collecting an extended candidate support set $ \widehat{\mathcal{A}}_i $ with $ |\widehat{\mathcal{A}}_i| \leq 2k $ (Step 4) and then finding the $ k $-sparse vector $ \as_{i+1} $ that best fits the measurements within the restricted support set $ \widehat{\mathcal{A}}_i $, i.e., the support set $ \as_{i+1} $ satisfies $ \mathcal{A}_{i+1} \triangleq \text{supp}(\as_{i+1}) \subseteq \widehat{\mathcal{A}}_i $ (Steps 5-6).

\begin{algorithm}
   \caption{Subspace Pursuit Algorithm}\label{algo: class}
 {\bfseries Input:} $\obs$, $\A$, $ k $, MaxIter.
 {\bfseries Output:} $ \hat{\as} \leftarrow \argmin_{\mathbf{v}: \text{supp}(\mathbf{v}) \subseteq \mathcal{A}_i} \vectornorm{\obs - \A \mathbf{v}}_2^2 $
\end{algorithm}



In \cite{foucart2010sparse}, Foucart improves the initial RIP conditions of SP algorithm, which we present here as a corollary:

\begin{corollary}[SP Iteration Invariant]\label{eq: CoSaMP2}
SP algorithm satisfies the following recursive formula:
\begin{align}
\vectornorm{\as_{i+1} - \xtrue}_2 \leq \rho \vectornorm{\as_{i} - \xtrue}_2 + c \vectornorm{\noise}_2,
\end{align} where $ c = \sqrt{\frac{2(1+3\delta_{3k}^2)}{1-\delta_{3k}}} + \frac{\sqrt{(1+3\delta_{3k}^2)(1+\delta_{2k})}}{1-\delta_{3k}} + \sqrt{3(1+\delta_{2k})} $ and $ \rho < 1 $ given that $ \delta_{3k} < 0.38427 $.
\end{corollary}

\subsection{Algorithm Description}
In this section, we expose \class algorithm, a Subspace Pursuit \cite{SP} variant, as a running example for our subsequent developments. We underline that norm constraints can be also incorporated into alternative state-of-the-art hard thresholding frameworks  \cite{recipes,SP,cosamp,Blumensath_iterativehard,Foucart_hardthresholding}. 

\begin{algorithm}
   \caption{The \class Algorithm}\label{algo: class}
    {\bfseries Input:} $\obs$, $\A$, $ \Delta_{\ell_0, \ell_1}(\sparsity, \tau) $, Tolerance, MaxIterations
 {\bfseries Output: $\signal_i$.}
\end{algorithm}
\vspace{-0.1cm}

The \class algorithm approximates $\bestsignal $ according to the optimization formulation \eqref{spexact} where $q=2$. We provide a pseudo-code of an example implementation of \class in Algorithm \ref{algo: class}. To complete the $i$-th iteration, \textsc{Clash} initially identifies a $2k$ extended support set $\widehat{\mathcal{A}}_i$ to explore via the {\it Active set expansion} step (Step 1)---the set $\widehat{\mathcal{A}}_i$ is constituted by the union of the support $\mathcal{A}_i$ of the current solution $\signal_i$ and an additional $k$-sparse support where the projected gradient onto $ \Delta_{\ell_0}(k) $ can make most impact on the loading vector, {\it complementary} to $\mathcal{A}_i$. Given $\widehat{\mathcal{A}}_i$, the {\it Greedy descent with least absolute shrinakge} step (Step 2) solves a least-squares problem over {\it $\ell_1$-norm constraint} to decrease the data error $f(\signal)$, restricted over the active support set $\widehat{\mathcal{A}}_i$. In sequence, we project the $2k$-sparse solution of Step 2 onto $ \Delta_{\ell_0}(k) $ to arbitrate the active support set via the {\it Combinatorial selection} step (Step 3). Finally, \class de-biases the result on the putative solution support using the {\it De-bias} step (Step 4).

\subsection{The CLASH Guarantees}

\class iterations satisfy the following worst-case guarantee:

\begin{theorem}\label{thm: iteration invariant}[Iteration invariant] Let $ \bestsignal $ be the true solution. Then, the $i$-th iterate $\signal_i$ of \class satisfies the following recursion \vspace{-0.1cm}
\begin{align}
\vectornorm{\signal_{i+1} - \bestsignal}_2 &\leq \rho \vectornorm{\signal_i - \bestsignal}_2 + c_1(\delta_{2\sparsity}, \delta_{3\sparsity})\vectornorm{\noise}_2, ~\text{where} \label{eq:59b} \\
c_1(\delta_{2\sparsity}, \delta_{3\sparsity}) &\triangleq \frac{1}{\sqrt{1-\delta_{2\sparsity}^2}}\Bigg(\sqrt{1+3\delta_{3\sparsity}^2}\Big(\sqrt{\frac{2(1+\delta_{3\sparsity})}{1-\delta_{3\sparsity}^2}} \nonumber \\  &+ \frac{\sqrt{1+\delta_{2\sparsity}}}{1-\delta_{3\sparsity}}\Big) + \sqrt{3(1+\delta_{2\sparsity})}\Bigg) + \frac{\sqrt{1+\delta_\sparsity}}{1-\delta_{2\sparsity}},  \vspace{-0.1cm}
\end{align}  and $\rho \triangleq \frac{\delta_{3\sparsity} + \delta_{2\sparsity}}{\sqrt{1-\delta_{2\sparsity}^2}} \sqrt{\frac{1+3\delta_{3\sparsity}^2}{1-\delta_{3\sparsity}^2}}$. Moreover, when $\delta_{3\sparsity} < 0.3658$, the iterations are contractive (i.e., $ \rho < 1 $).
\end{theorem}

A detailed proof of Theorem \ref{thm: iteration invariant} can be found in \cite{clash}.
%
Theorem \ref{thm: iteration invariant} shows that the isometry requirements of \class are competitive with those of mainstream hard thresholding methods, such as SP, even though \class incorporates the $\ell_1$-norm constraints---furthermore, we observe improved signal reconstruction performance compared to these methods, as shown in the Experiments section.

\section{Experiments}\label{sec: exp}
In this section, we provide experimental results to demonstrate the performances of the {\it \textsc{GAME}} and \class Algorithms.

\subsection{Performance of the $\ell_\infty$ {\it \textsc{GAME}} algorithm}
In this experiment, we fix $\n=1000$, $\mm=200$ and $k=20$, and generate a $200\times 1000$ Gaussian matrix $\A$. Each experiment is repeated independently $50$ times. We compare the performance of the $\ell_\infty$ {\it \textsc{GAME}} algorithm, which approximately solves the non-convex problem \begin{equation}\label{spdantz}
\begin{aligned}
	& \underset{\as\in\Delta_{\ell_0, \ell_1}(k,\tau)}{\text{minimize}}
	& & \|\A^\top\A\as-\A^\top\f\|_\infty	
\end{aligned}
\end{equation}  with state-of-the-art Dantzig Selector solvers \cite{dantzig,homotopy} that solve linear optimization
\begin{equation}\label{spdantz2}
\begin{aligned}
	& \underset{\as\in\Delta_{\ell_1}(\tau)}{\text{minimize}}
	& & \|\A^\top\A\as-\A^\top\f\|_\infty	
\end{aligned}
\end{equation}

The compressive measurements were generated in the presence of white Gaussian noise. The noise vector consists of $\mm$ iid ${\cal N}(0,\sigma^2)$ elements, where $\sigma$ ranges from $10^{-3.5}$ to $10^{-0.5}$. Figure~\ref{fig_dantz} compares the data-domain $\ell_2$-error ($\|\sas-\has\|_2/\|\sas\|_2$) of the {\it \textsc{GAME}} algorithm with the error of $\ell_1$-magic algorithm \cite{l1magic} and the Homotopy algorithm \cite{salman} which are state-of-the-art Dantzig Selector optimizers. As illustrated in Figure~\ref{fig_dantz}, as $\sigma$ increases to $10^{-3}$, the {\it \textsc{GAME}} algorithm outperforms the $\ell_1$-magic and Homotopy algorithms.


\begin{figure}[t]
{
\centerline{\includegraphics[scale=0.3] {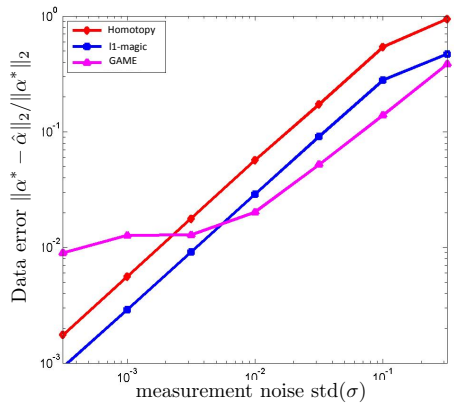}}
   \label{fig_dantz_2}
 }

\caption{Signal approximation experiments with $\ell_1$-magic, Homotopy, and {\it \textsc{GAME}} algorithms. The measurement noise standard deviation ranges from $10^{-3.5}$ to $10^{-0.5}$, and the approximation error is measured as $\|\sas-\has\|_2/$.}\label{fig_dantz}
\end{figure}

\subsection{Performance of \textsc{Clash} Algorithm}

\textbf{Noise resilience:} We generate random realizations of the model $ \obs = \A \bestsignal $ for $ \dimension = 1000 $, $\numsam = 305 $ and $\sparsity = 115 $ where $ \sparsity $ is known a-priori and $ \bestsignal $ admits the simple sparsity model. We construct $\bestsignal$ as a $k$-spare vector with iid ${\cal N}(0,1)$ elements with $\|\bestsignal\|_2 = 1$. We repeat the same experiment independently for $50$ Monte-Carlo iterations. In this experiment, we examine the signal recovery performance of \class compared to the following state-of-the-art methods: $ i) $ Lasso (\ref{opt:01}) as a projected gradient method, $ ii) $ Basis Pursuit \cite{BPDN} using SPGL1 implementation \cite{BergFriedlander2008} and, $ iii) $ Subspace Pursuit \cite{SP}. We test the recovery performance of the aforementioned methods for various noise standard deviations -- the empirical results are depicted in Figure \ref{fig:clash}. We observe that the combination of hard thresholding with norm constraints significantly {\it improves} the signal recovery performance over both convex- and combinatorial-based approaches.

\begin{figure}[t]
{
\centerline{\includegraphics[scale=0.35] {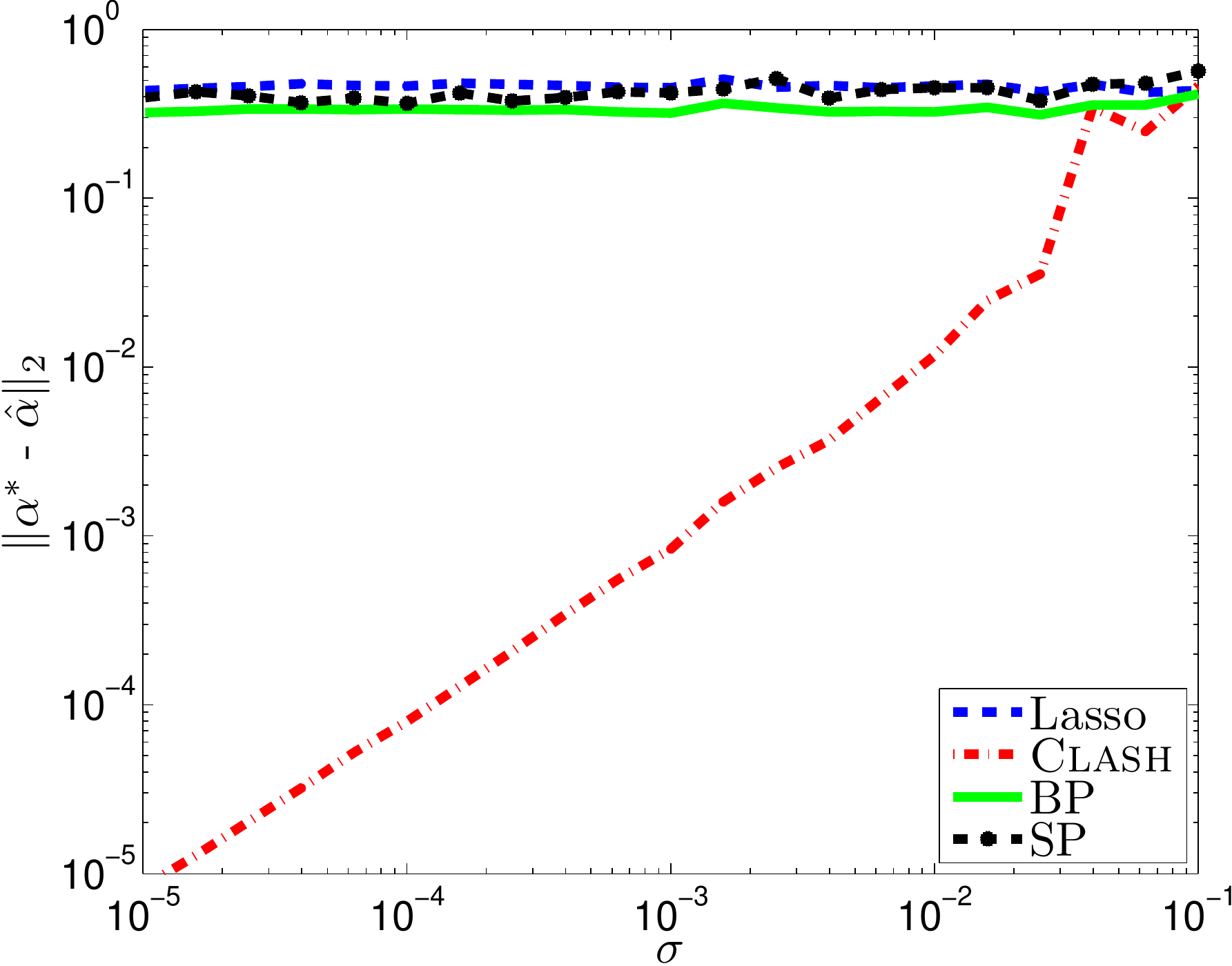}}
}
\caption{Signal approximation experiments with \textsc{Clash}, Lasso, and BP algorithms. The measurement noise standard deviation ranges from $10^{-5}$ to $10^{-1}$, and the approximation error is measured as $\|\sas-\has\|_2$.}{\label{fig:clash}}
\end{figure}

{\bf Improved recovery using} \textsc{Clash}{\bf :} We generate random realizations of the model $ \obs = \A \bestsignal + \mathbf{n} $ for $ \dimension = 500 $, $\numsam = 160 $ and $\sparsity = \lbrace 57, 62 \rbrace $ for the noisy and the noiseless case respectively, where $ \sparsity $ is known a-priori. We construct $\bestsignal$ as a $k$-spare vector with iid ${\cal N}(0,1)$ elements with $\|\bestsignal\|_2 = 1$.  In the noisy case, we assume $ \|\mathbf{n}\|_2 = 0.05 $. We perform $500$ independent Monte-Carlo iterations. We then sweep $\tau$ and then examine the signal recovery performance of \class compared to the same methods above. Note that, if $ \tau $ is large, norm constraints have no impact in recovery and \class must admit identical performance to SP.

Figure \ref{clash_recovery} illustrates that the combination of hard thresholding with norm constraints can {\it improve} the signal recovery performance significantly over convex-only and hard thresholding-only methods.  \class perfectly recovers the signal when the regularization parameter is close to $\|\alpha^\ast\|_1 $. When $\tau \ll \vectornorm{\bestsignal}_1$ or $\tau \gg \vectornorm{\bestsignal}_1$, the performance degrades.

\begin{figure}[t]
\begin{tabular}{cc}
\centerline{\subfigure{\includegraphics[scale=0.34]{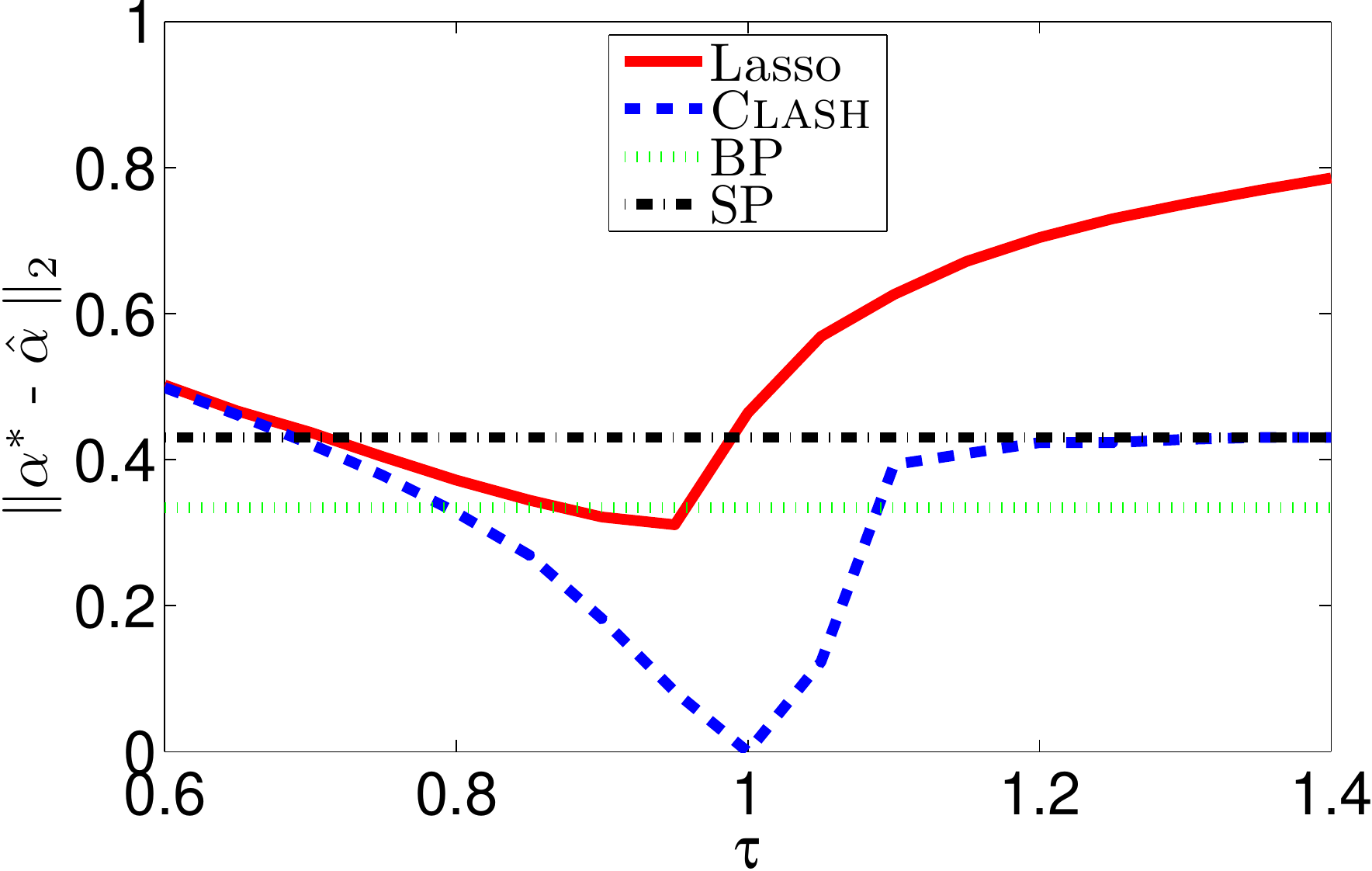}}
\hfill
\subfigure{\includegraphics[scale=0.34]{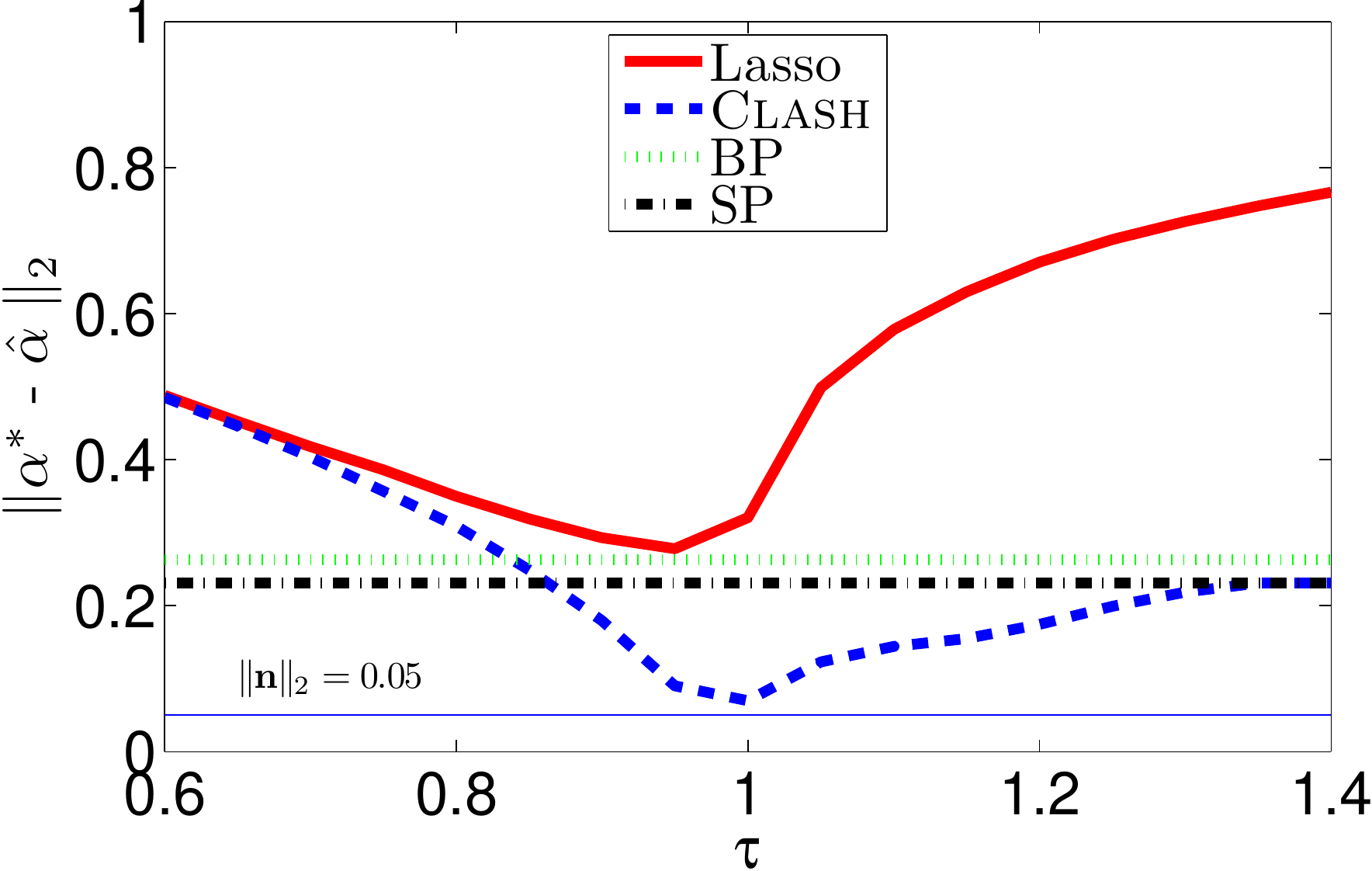}}}
\end{tabular}
\caption{Improved signal recovery using \textsc{Clash}.}{\label{clash_recovery}}
\end{figure} 
\section{Conclusions}\label{sec: conclusions}
We discussed two sparse recovery algorithms that explicitly leverage convex $\ell_1$ and non-convex $\ell_0$ priors jointly. While the $\ell_1$ prior is conventionally motivated as the ``convexification'' of the $\ell_0$ prior, we saw that this interpretation is incomplete: it actually is a convexification of the $\ell_0$-constrained set with a maximum scale. We also discovered that the interplay of these two---seemingly related---priors could lead to not only strong theoretical recovery guarantees from weaker assumptions than commonly used in sparse recovery, but also improved empirical performance over the existing solvers. To obtain our results, we reviewed some important topics from game theory, convex and combinatorial optimization literature. We believe that understanding and exploiting the interplay of such convex and non-convex priors could lead to radically new, scalable regression approaches, which can leverage decades of work in diverse theoretical disciplines.

\section*{Acknowledgements}
VC and AK's work was supported in part by the European Commission under Grant MIRG-268398, ERC Future Proof, DARPA KeCoM program $\#$ 11-DARPA-1055, and SNF $200021$-$132548$ grants. VC also would like to acknowledge Rice University for his Faculty Fellowship. SJ thanks Robert Calderbank and Rob Schapire for providing insightful comments.

\bibliographystyle{IEEEtran}
\small{
\bibliography{bibs/bib,bibs/bib2,bibs/bibgame,bibs/allbib}}



\endgroup


\end{document}